\documentclass[letterpaper, 10 pt, conference]{ieeeconf}  
\IEEEoverridecommandlockouts                              
\usepackage{balance}
\usepackage{color}
\usepackage{cite}
\usepackage{verbatim}
\usepackage{hyperref}

\usepackage{soul}
\usepackage[utf8]{inputenc}
\usepackage[T1]{fontenc}
\usepackage{xcolor}
\usepackage{graphicx}
\usepackage{flushend}
\usepackage{algorithm}
\usepackage{algpseudocode}
\algrenewcommand\algorithmicrequire{\textbf{Input:}}
\algrenewcommand\algorithmicensure{\textbf{Output:}}
\definecolor{cambridgeblue}{rgb}{0.64, 0.76, 0.68}
\definecolor{ceil}{rgb}{0.57, 0.63, 0.81}
\definecolor{celestialblue}{rgb}{0.29, 0.59, 0.82}
\newcommand{\CommentState}[1]{\Statex\hspace{\algorithmicindent}{\color{celestialblue}// #1}}

\usepackage{amsmath,amsfonts,amssymb}
\usepackage{bm,fixmath} 
\usepackage{accents}

\newtheorem{definition}{Definition}
\newtheorem{remark}{Remark}

\newtheorem{proposition}{Proposition}
\newtheorem{theorem}{Theorem}
\newtheorem{assumption}{\bfseries Assumption}
\newtheorem{lemma}{Lemma}

\hypersetup{
    colorlinks = true,
    linkcolor = [rgb]{0,0,1},
    anchorcolor = [rgb]{0,0,1},
    citecolor = [rgb]{0.9,0.5,0},
    filecolor = [rgb]{0,0,1},
    urlcolor = [rgb]{0.9,0.5,0},
    bookmarksopen = true,
    bookmarksnumbered = true,
    breaklinks = true,
    linktocpage,
    colorlinks = true,
    linkcolor = [rgb]{0.2,0.6,0.2},
    urlcolor  = [rgb]{0,0,1},
    citecolor = [rgb]{0.9,0.5,0},
    anchorcolor = [rgb]{0.2,0.6,0.2},
}

\newenvironment{list4}{
\begin{list}{$\bullet$}{%
    \setlength{\itemsep}{0.05cm}
    \setlength{\labelsep}{0.2cm}
    \setlength{\labelwidth}{0.3cm}
    \setlength{\parsep}{0in} 
    \setlength{\parskip}{0in}
    \setlength{\topsep}{0in} 
    \setlength{\partopsep}{0in}
    \setlength{\leftmargin}{0.18in}}}
{\end{list}}

\DeclareMathOperator*{\argmin}{arg\,min}
\DeclareMathOperator{\proj}{proj}

\newcommand{\norm}[1]{\left\lVert#1\right\rVert}
\newcommand{\ev}[1]{\mathbb{E}\left[#1\right]}

\newcommand{\R}{\mathbb{R}}

\newcommand{\C}{\mathcal{C}}

\newcommand{\e}{\mathbold{e}}
\newcommand{\x}{\mathbold{x}}

\title{\LARGE \bf
Distributed Optimization and Learning for Automated Stepsize Selection \\ with Finite Time Coordination
}

\author{Apostolos~I.~Rikos, Nicola~Bastianello, Themistoklis~Charalambous and
 Karl~H.~Johansson
    \thanks{Apostolos~I.~Rikos is with the Artificial Intelligence Thrust of the Information Hub, The Hong Kong University of Science and Technology (Guangzhou), Guangzhou, China. 
    He is also affiliated with the Department of Computer Science and Engineering, The Hong Kong University of Science and Technology, Clear Water Bay, Hong Kong. E-mail: {\tt~apostolosr@hkust-gz.edu.cn}.}
    \thanks{N.~Bastianello and K.~H.~Johansson are with the School of Electrical Engineering and Computer Science, and Digital Futures, KTH Royal Institute of Technology, Sweden. They are also affiliated with Digital Futures, SE-100 44 Stockholm, Sweden. E-mails: {\tt~\{nicolba | kallej\}@kth.se}.}
    \thanks{Themistoklis~Charalambous is with the Department of Electrical and Computer Engineering, University of Cyprus, 1678 Nicosia, Cyprus. 
    He is also a Visiting Professor with the Department of Electrical Engineering and Automation, School of Electrical Engineering, Aalto University. 
    E-mail:{\tt~charalambous.themistoklis@ucy.ac.cy}.}
\thanks{The work of A.I.R. was supported by the Guangzhou-HKUST(GZ) Joint Funding Scheme (Grant No. 2025A03J3960). 
The work of N.B. and K.H.J. was partially supported by the European Union’s Horizon Research and Innovation Actions programme under grant agreement No. 101070162, and partially by Swedish Research Council Distinguished Professor Grant 2017-01078 Knut and Alice Wallenberg Foundation Wallenberg Scholar Grant. 
The work of T.C. was partially funded by MINERVA, a European Research Council (ERC) project funded under the European Union's Horizon 2022 research and innovation programme (Grant agreement No. 101044629).}
}

\begin{document}

\maketitle
\thispagestyle{plain}
\pagestyle{plain}

\begin{abstract}
%
Distributed optimization and learning algorithms are designed to operate over large scale networks enabling processing of vast amounts of data effectively and efficiently. 
One of the main challenges for ensuring a smooth learning process in gradient-based methods is the appropriate selection of a learning stepsize. 
Most current distributed approaches let individual nodes adapt their stepsizes locally. 
However, this may introduce stepsize heterogeneity in the network, thus disrupting the learning process and potentially leading to divergence. 
In this paper, we propose a distributed learning algorithm that incorporates a novel mechanism for automating stepsize selection among nodes. 
Our main idea relies on implementing a finite time coordination algorithm for eliminating stepsize heterogeneity among nodes. 
We analyze the operation of our algorithm and we establish its convergence to the optimal solution. 
We conclude our paper with numerical simulations for a linear regression problem, showcasing that eliminating stepsize heterogeneity enhances convergence speed and accuracy against current approaches. 
\end{abstract}

\begin{keywords}
Distributed optimization and learning, automated stepsize selection, finite-time coordination. 
\end{keywords}

\section{Introduction}\label{sec:intro}

In today's interconnected world, networks of devices (or nodes) are increasingly deployed to collect and process vast amounts of data from their environments. 
Efficiently processing this data is essential for solving complex problems and making informed decisions. 
While centralized approaches (where nodes transmit their data to a single entity for processing) have been widely used, they face significant limitations such as high computational demands on the central entity, risks to data privacy, vulnerability to single points of failure, and scalability challenges \cite{2021_Meng_Yang_Johansson_Book}. 
As a result, there is growing interest in distributed algorithms that enable interconnected nodes to make collective decisions by locally processing their stored dataset and communicating with other nodes only the result of their computation. 
This shift has positioned the distributed optimization and learning problem as a pivotal area of research in the scientific community \cite{2021_Meng_Yang_Johansson_Book, 2018_liu_nedic_distr_opt_survey, 2025_doostmohammadian_rikos_Johansson_survey}. 

The problem of distributed optimization and learning has become increasingly important in recent years due to the rise of large-scale machine learning \cite{2020:Nedich}, control \cite{SEYBOTH:2013}, and other data-driven applications \cite{2018:Stich_Jaggi} that involve the processing of massive amounts of data. 
Distributed optimization and learning algorithms require nodes to perform two fundamental operations, processing locally stored data (typically through gradient descent), and communicating the results of this processing with other nodes (by exchanging either real-valued messages \cite{2021:Nedic_PushPull, 2022:Jiang_Charalambous}, or quantized messages \cite{2023:Bastianello_Rikos_Johansson, 2023:Rikos_Johan_IFAC}). 
A critical parameter in gradient-based local processing is the learning stepsize as it determines the magnitude of parameter updates based on the observed gradient. 
Improper stepsize selection may result in slow convergence or even divergence. 
Therefore, selecting an appropriate stepsize is a critical challenge in distributed optimization and learning. 

While centralized optimization leverages various techniques for adaptive stepsize selection (see, for example, AdaGrad~\cite{JMLR:v12:duchi11a} and Adam~\cite{ADAM:2015}), applying these methods in a distributed setting presents significant difficulties. 
In most existing distributed optimization and learning algorithms, the stepsize is typically chosen as a fixed or a diminishing value (see, for example \cite{Nedic:2024} and references therein). 
This choice depends on factors such as the network size, topological structure, and the Lipschitz constant of the global objective function, all of which are challenging or even infeasible to obtain in a distributed fashion. 
On another note, a proper stepsize selection may require extensive manual tuning from users to ensure effective algorithm convergence. 
However, a well-tuned stepsize for one scenario may not generalize well to another, limiting its adaptability. 
Another challenge arises when nodes independently adapt their stepsizes based on local data or partial views of the global objective function. 
This introduces stepsize heterogeneity across the network and can cause algorithmic divergence. 

For addressing the aforementioned challenges, recent research focused on developing distributed optimization and learning algorithms that incorporate stepsize adaptation techniques. 
For instance, the work in \cite{2012_Jadbabaie_Line_Search} introduced a backtracking line-search method for adapting stepsizes in a distributed optimization setting. 
However, this approach incurs significant computational overhead as it requires executing additional subroutines and performing multiple gradient evaluations. 
The work in \cite{gao2022achieving} employed the Barzilai-Borwein stepsize in a gradient-tracking-based algorithm, but this method depends on prior knowledge of the global Lipschitz constant and strong convexity coefficient (which are often impractical to obtain in distributed environments). 
Another approach proposed in \cite{Michailidis2022dadam}, introduced an automated stepsize selection mechanism for distributed learning algorithms. 
However, this method still requires manual tuning of learning-rate parameters, which limits its adaptability across different scenarios.
The works in \cite{liggett2022distributed, Yongqiang2024AutomatedStepsizes} proposed distributed learning algorithms that successfully automate stepsize selection without requiring manual adjustments from users. 
Nevertheless, their algorithms operate only for undirected or balanced directed graphs. 
Additionally, their operation relies on the implementation of dynamic consensus on the stepsizes, leading to possible instabilities and instances where convergence is not guaranteed. 
This is because dynamic consensus involves only a single communication step, which reduces stepsize heterogeneity but fails to eliminate it entirely, leaving room for misalignment in gradient descent directions and potential divergence in distributed optimization processes. 
Finally, the work in \cite{Nedic:2024} proposes distributed optimization algorithms that leverage a designed mixing matrix and time-varying uncoordinated stepsizes to enable privacy among nodes. 
However, these algorithms also suffer from the same shortcomings as the previously mentioned approaches. 
Specifically, they manage to reduce stepsize heterogeneity without eliminating it entirely and are restricted to operating only over undirected graphs.
Overall, the aforementioned works successfully reduce stepsize heterogeneity among nodes, without however eliminating it completely. 
Furthermore, either they impose significant computational overhead or are limited to operating over undirected or balanced directed graphs. 
To the best of the authors' knowledge, designing distributed learning algorithms with integrated mechanisms for automated stepsize selection that completely eliminate stepsize heterogeneity (thereby enhancing convergence stability and efficiency) while operating in directed, unbalanced communication networks remains an open problem in the literature. 

\textbf{Main Contributions.} 
Motivated by the aforementioned gap in the literature, we propose a distributed learning algorithm that incorporates a novel mechanism for automated stepsize selection while operating in directed, unbalanced communication networks. 
Our algorithm eliminates stepsize heterogeneity among nodes, thereby enhancing convergence stability and efficiency. 
Our approach is inspired from a recently proposed stepsize automation method for centralized optimization in \cite{Mishchenko2020Adaptive} and the proposed method in \cite{liggett2022distributed}. 
It is important to note that directly applying the centralized stepsize automation technique from \cite{Mishchenko2020Adaptive} to distributed settings inevitably results in stepsize heterogeneity among agents, which may lead to divergence. 
Therefore, the extension of centralized stepsize automation techniques to distributed optimization settings serves as an additional motivation for this work. 
Our main contributions are the following. 
\\ 
\textbf{A.} We propose a distributed learning algorithm that incorporates a novel mechanism for automated stepsize selection while operating in directed, unbalanced communication networks (see Algorithm~\ref{Algorithm_real}). 
Our algorithm’s operation relies on the implementation of a finite time coordination algorithm for eliminating stepsize heterogeneity among nodes. 
This strategy enhances convergence stability and efficiency by ensuring synchronized updates across all nodes. 
\\ 
\textbf{B.} 
We establish our algorithm's convergence to the optimal solution for the case where the local cost functions of nodes are strongly convex and smooth (see Theorem~\ref{theorem_convergence_Alg1}).

%
%
%
%
\section{Notation and Background}\label{preliminaries}

We denote the sets of real, rational, natural, integer, nonnegative real, and nonnegative integer numbers by $\mathbb{R}$, $\mathbb{Q}$, $\mathbb{N}$, $\mathbb{Z}$, $\mathbb{R}_{+}$, and $\mathbb{Z}_+$, respectively. 
For a differentiable function $f$ its gradient is defined as $\nabla f$, and its stochastic gradient as $\hat{\nabla} f$. 
The expected value of a random variable $X$ is denoted by $\ev{X}$.
For a set $C$, the projection of a point $y$ onto $C$ is defined as $\proj_\C(y) = \argmin_{x \in C} \norm{x - y}^2$. 

\textbf{Graph-Theoretic Framework.} 
We consider a network comprising $N$ nodes ($n \geq 2$), where communication is restricted to immediate neighbors. 
The communication topology is represented by a directed graph (digraph) $\mathcal{G}_d = (\mathcal{V}, \mathcal{E})$, where 
$\mathcal{V} = \{v_1, v_2, \dots, v_n\}$ is the set of nodes, $\mathcal{E} \subseteq \mathcal{V} \times \mathcal{V} - \{ (v_i, v_i) \mid v_i \in \mathcal{V} \}$ is the set of edges (excluding self-edges). 
A directed edge from node $v_i$ to node $v_l$, denoted as $m_{li} \triangleq (v_l, v_i) \in \mathcal{E}$, indicates that node $v_l$ can receive information from node $v_i$ (unidirectional communication). 
A digraph $\mathcal{G}_d = (\mathcal{V}, \mathcal{E})$ being \textit{strongly connected} means that for any pair of distinct nodes $v_l, v_i \in \mathcal{V}$, there exists a directed \textit{path} from $v_i$ to $v_l$. 
For each node $v_i$, we define the sets: in-neighbors as $\mathcal{N}_i^- = \{ v_j \in \mathcal{V} \mid (v_i,v_j)\in \mathcal{E}\}$, out-neighbors as $\mathcal{N}_i^+ = \{ v_l \in \mathcal{V} \mid (v_l,v_i)\in \mathcal{E}\}$, in-degree as $\mathcal{D}_i^- = | \mathcal{N}_i^- |$, out-degree as $\mathcal{D}_i^+ = | \mathcal{N}_i^+ |$. 
The in-neighbors of $v_i$ are nodes that can directly transmit information to $v_i$, and the out-neighbors are nodes that can directly receive information from $v_i$. 
The in-degree and out-degree represent the cardinalities of these sets, respectively. 



\section{Preliminaries on Distributed Coordination}\label{prel_distr_coord}

\subsection{Ratio Consensus}\label{subsec_ratio_cons}

The work in \cite{2010:christoforos} introduced a novel algorithm called ratio consensus to solve the average consensus problem in directed graphs. 
In this algorithm, each node $v_i$ assigns positive weights to its self-link and outgoing links, resulting in a column stochastic weight matrix $P$ (note that $P$ is not necessarily row stochastic). 
By utilizing this weight matrix $P$ and executing two iterations with specific initial conditions, the algorithm achieves average consensus. 
This algorithm is detailed below in Proposition~\ref{lemma_christoforos}. 
It prescribes a particular weighting scheme for each link, assuming each node knows its out-degree. 
It is worth noting that the algorithm is applicable to any weight configuration conforming to the graph topology and forming a primitive column stochastic weight matrix. 
The algorithm in Proposition~\ref{lemma_christoforos} below achieves the precise average in an \textit{asymptotic} manner. 

\begin{proposition}[\hspace{-0.09cm} \cite{2010:christoforos}] 
\label{lemma_christoforos} 
Consider a strongly connected directed graph $\mathcal{G}_d = (\mathcal{V}, \mathcal{E})$. 
Let $y_i^{[t]}$ and $x_i^{[t]}$ (for all $v_i \in \mathcal{V}$ and $t = 0,1,2,\ldots$) be the results of the iterations 
\begin{subequations}\label{eq:1}
\begin{align}
y_i^{[t+1]} = p_{ii} y_i^{[t]} + \sum_{v_j \in \mathcal{N}^{-}_i} p_{ij} y_j^{[t]} , \label{subeq:1} \\
x_i^{[t+1]} = p_{ii} x_i^{[t]} + \sum_{v_j \in \mathcal{N}^{-}_i} p_{ij} x_j^{[t]} , \label{subeq:2}
\end{align}
\end{subequations}
where $p_{ij} = \frac{1}{1 + \mathcal{D}_j^+}$ for $v_i \in \mathcal{N}_j^+ \cup { v_j }$ (zeros otherwise), with initial conditions $x_i^{[0]} = \bm{1}$ and $y_i^{[0]}$ being the initial state of each node. 
Then, the average consensus problem is asymptotically solved as 
$
\lim_{k \rightarrow \infty} \mu_i^{[t]}=\frac{\sum_{v_i \in \mathcal{V}} y_i^{[0]}}{|\mathcal{V}|} , \ \forall v_i \in \mathcal{V}, 
$
where
$
\mu_i^{[t]} = y_i^{[t]} / x_i^{[t]}. 
$
\end{proposition}

\subsection{Finite-Time Exact Ratio Consensus (FTERC)}\label{subsec_FTERC}


We now introduce a distributed algorithm (developed in \cite{themisCDC:2013} and \cite{themisTCNS:2015}) that enables each node to calculate the exact average within a minimal number of time steps using only its local observations. 
This algorithm builds upon the algorithm presented in Proposition~\ref{lemma_christoforos}, allowing each node to compute $\mu_i \triangleq \lim_{t \rightarrow \infty} \mu_i^{[t]}=\frac{\sum_{v_i \in \mathcal{V}} y_i^{[0]}}{|\mathcal{V}|} , \ \forall v_i \in \mathcal{V}$ in the \emph{minimum} possible number of iterations. 


\begin{definition}[Minimal Polynomial of a Matrix Pair] 
For a matrix pair $[P,e_i^\top]$, the associated minimal polynomial, denoted by
$q_i(s)=s^{[M_i+1]}+\sum_{j=0}^{M_i} \alpha^{(i)}_j s^{[j]}$, is the unique monic polynomial of lowest degree $M_i+1$ that
satisfies $e^\top_i q_i(P)=0$.
\end{definition}

For the iteration $w_i^{[t+1]} = P w_i^{[t]}$ with weight matrix $P$, it can be demonstrated (using methods similar to those in \cite{2009:Ye}) that
\begin{equation}\label{regression}
\sum _{j=0}^{M_i+1} \alpha^{(i)}_j w_i^{[t+j]}=0, \quad \forall t \in \mathbb{Z}_+ , 
\end{equation}
where $\alpha^{(i)}_{M_j+1}=1$. 
Let $W_i(z) \triangleq \mathcal{Z}(w_i^{[t]})$ denote the $z$-transform of $w_i^{[t]}$. 
Applying \eqref{regression} and leveraging the time-shift property of the $z-$transform yields (as shown in \cite{2009:Ye,2013:Ye})
\begin{equation} \label{ztranform}
W_i(z)=\frac{\sum _{j=1}^{M_i+1} \alpha^{(i)}_j \sum _{\ell=0}^{j-1}w_i^{[\ell]} z^{[j-\ell]}}{q_i(z)} ,
\end{equation}
where $q_i(z)$ represents the minimal polynomial of $[P,e^\top_i]$. 
For a strongly connected network, $q_i(z)$ has no unstable poles except for one at $1$. 
This allows us to define the polynomial:
\begin{align}\label{beta}
p_i(z) \triangleq \frac{q_i(z)}{z-1} \triangleq \sum_{j=0}^{M_i} \beta^{(i)}_j z^{[j]} .
\end{align}


Applying the final value theorem \cite{2009:Ye,2013:Ye}, we obtain:
\begin{subequations}\label{eq:phi}
\begin{align}
\phi_{y}(i) &= \lim_{t\rightarrow \infty}y_i^{[t]} = \lim_{z\rightarrow 1}(z-1)Y_i(z) = \frac{y_{M_i}^\top {\bm \beta}_i}{\bm{1}^\top {\bm \beta}_i} , \label{phi:1} \\
\phi_{x}(i) &= \lim_{t\rightarrow \infty}x_i^{[t]} = \lim_{z\rightarrow 1}(z-1)X_i(z) = \frac{x_{M_i}^\top {\bm \beta}_i}{\bm{1}^\top {\bm \beta}_i} , \label{phi:2}
\end{align}
\end{subequations}
where $y_{M_i}^\top = [y_i^{[0]}, y_i^{[1]}, \ldots, y_i^{[M_j]}]$ and $x_{M_i}^\top = [x_i^{[0]}, x_i^{[1]}, \ldots, x_i^{{[M_j]}}]$ represent the vectors of the first $M_i + 1$ values of $y_i^{[t]}$ and $x_i^{[t]}$ respectively, and ${\bm \beta}_i$ denotes the coefficient vector of the polynomial $p_i(z)$.

Let us consider vectors of $2t+1$ consecutive discrete-time values at node $v_i$, defined as
\begin{align*}
y_{2t}^\top &= [y_i^{[0]}, y_i^{[1]}, \ldots, y_i^{[2t]} ], \\
x_{2t}^\top &= [x_i^{[0]}, x_i^{[1]}, \ldots, x_i^{[2t]} ],
\end{align*}
corresponding to the two iterations $y_i^{[t]}$ and $x_i^{[t]}$ at node $v_i$ (as specified in equations \eqref{subeq:1} and \eqref{subeq:2}, respectively). 
We define their associated Hankel matrices as follows:
\begin{subequations}\label{hankel_def}
\begin{align}\label{eq:hankel_def}
\Gamma\{y_{2t}^\top\} &\triangleq \begin{bmatrix}
y_i^{[0]} & y_i^{[1]} & \ldots & y_i^{[t]} \\
y_i^{[1]} & y_i^{[2]} & \ldots & y_i^{[t+1]} \\
\vdots & \vdots & \ddots & \vdots \\
y_i^{[t]} & y_i^{[t+1]} & \ldots & y_{i}^{[2t]}
\end{bmatrix}, \\
\Gamma\{x_{2t}^\top\} &\triangleq \begin{bmatrix}
x_i^{[0]} & x_i^{[1]} & \ldots & x_i^{[t]} \\
x_i^{[1]} & x_i^{[2]} & \ldots & x_i^{[t+1]} \\
\vdots & \vdots & \ddots & \vdots \\
x_i^{[t]} & x_i^{[t+1]} & \ldots & x_{i}^{[2t]}
\end{bmatrix}.
\end{align}
\end{subequations}
Additionally, we introduce vectors of successive differences for $y_i^{[t]}$ and $x_i^{[t]}$: 
\begin{align*}
\overline{y}_{2t}^\top &= [y_i^{[1]}-y_i^{[0]}, \ldots, y_i^{[2t+1]}-y_i^{[2t]}], \\
\overline{x}_{2t}^\top &= [x_i^{[1]}-x_i^{[0]}, \ldots, x_i^{[2t+1]}-x_i^{[2t]}]. 
\end{align*}
As demonstrated in \cite{2013:Ye}, the vector ${\bm \beta}_i$ can be computed as the kernel of the first defective Hankel matrices $\Gamma\{\overline{y}_{2t}^\top\}$ and $\Gamma\{\overline{x}_{2t}^\top\}$ for arbitrary initial conditions $y_0$ and $x_0$. 
Specifically, ${\bm \beta}_i$ can be calculated as the normalized kernel 
$$ 
{\bm \beta}_i = \begin{bmatrix} {\bm \beta}_i^{[0]} & {\bm \beta}_i^{[1]} & \ldots & {\bm \beta}_i^{[M_j-1]} & 1 \end{bmatrix}^\top 
$$
of the first defective Hankel matrix $\Gamma\{\overline{y}_{2t}^\top\}$, except for a set of initial conditions with Lebesgue measure zero.


We now present Lemma~\ref{lem:main} (originally proposed in \cite{themisCDC:2013}). 
It establishes that in a strongly connected digraph, nodes can distributively compute the \textit{exact} average $\mu$ within a finite number of steps.


\begin{lemma}[\hspace{-0.001cm}\cite{themisCDC:2013}]\label{lem:main}
Consider a strongly connected graph $\mathcal{G}_d = (\mathcal{V}, \mathcal{E})$. 
Let $y_i^{[t]}$ and $x_i^{[t]}$ (for all $v_i \in \mathcal{V}$ and $t=0,1,2,\ldots$) be the outcomes of iterations \eqref{subeq:1} and \eqref{subeq:2}, respectively, where $P = [p_{ij}] \in \mathbb{R}_{+}^{n \times n}$ represents any set of weights that conform to the graph structure and constitute a primitive column stochastic weight matrix. 
The average consensus solution can be obtained distributively in finite time at each node $v_i$ by computing:
\begin{align}
\mu_i \triangleq \lim_{t \rightarrow \infty} \frac{y_i^{[t]}}{x_i^{[t]}} = \frac{\phi_y (i)}{\phi_x (i)} = \frac{y_{M_i}^\top{\bm \beta}_i}{x_{M_i}^\top{\bm \beta}_i},
\end{align}
where:
$\phi_y (i)$ and $\phi_x (i)$ are defined by equations \eqref{phi:1} and \eqref{phi:2}, respectively, and ${\bm \beta}_i$ is the coefficient vector as defined in \eqref{beta}. 
\end{lemma}

Lemma~\ref{lem:main} establishes that in a strongly connected digraph, the average consensus can be determined through the ratio of the asymptotic values of two parallel iterations. 
Specifically, iteration \eqref{subeq:1} with initial condition $y^{[0]}$, and iteration \eqref{subeq:2} with initial condition $x^{[0]} = \bm{1}$. 
It is important to emphasize the choice of $x^{[0]} = \bm{1}$ as an initial condition for the second iteration. 
This initialization does not fall within the Lebesgue measure zero set of matrix $P$ (as defined in Proposition~\ref{lemma_christoforos}), thereby ensuring the validity and generality of the approach. 


\section{Problem Formulation}\label{probForm}

Let us consider a strongly connected digraph $\mathcal{G}_d = (\mathcal{V}, \mathcal{E})$ with $N$ nodes. 
In this paper, we aim to develop a distributed algorithm that solves the following \textit{distributed optimization} problem 
\begin{subequations}\label{optim_prob}
\begin{align}
\min_{\x \in \R^{n N}}~ & F(\x) \equiv \sum_{i=1}^N f_i(x_i), \label{Global_cost_function}  \\
\text{s.t.}~ & \x \in C \equiv \{\x \in \R^{n N} | \ x_i = x_j, \forall v_i, v_j, \in \mathcal{V} \}, \label{constr_same_x}
\end{align} 
\end{subequations}
where $\x = [x_1^\top, \ldots, x_N^\top]^\top$ represents the concatenation of local variables of the $N$ nodes, and $f_i$ denotes the local cost of each node $v_i$ (only known to node $v_i$). 
In our problem we are interested in \textit{finite-sum local costs} of the form
\begin{equation}\label{local_cost_func}
f_i(x_i) = \frac{1}{m_i} \sum_{h=1}^{m_i} \ell_i^h(x_i; \xi_i^h),  
\end{equation}
where, $m_i$ represents the number of data samples available at node $v_i$, $\ell_i^h$ is a loss function associated with the $h$-th data sample at node $v_i$, and $\xi_i^h$ denotes the $h$-th data sample. 
The formulation in \eqref{local_cost_func} is particularly relevant in machine learning scenarios (e.g., empirical risk minimization, federated learning, distributed regression). 

\subsection*{Operational Assumptions}

For solving problem \eqref{optim_prob} we introduce the following assumptions that hold throughout our paper. 

\begin{assumption}\label{str_conn}
The network is modeled as a \textit{strongly connected} digraph $\mathcal{G}_d = (\mathcal{V}, \mathcal{E})$. 
\end{assumption}

\begin{assumption}\label{lipsch_str_conv} 
    For every node $v_i$ its local cost function $f_i$ is smooth with constant $L_i$, strongly convex with constant $\mu_i$. 
    As a consequence, the global cost function $F$ (see e.g., \eqref{Global_cost_function}) has Lipschitz smoothness constant denoted by $L$ for which it holds $L \leq \sum_i L_i$, and strong convexity constant $\mu$ for which it holds $\mu \geq \min_i \mu_i$. 
\end{assumption} 


\begin{assumption}\label{bounded_grad} 
The stochastic gradient $\hat{\nabla} f_i (x_i)$ of every node $v_i$ satisfies $\ev{\hat{\nabla} f_i (x_i)} = \nabla f_i (x_i)$, and $\ev{\norm{\hat{\nabla} f_i (x_i) - \nabla f_i (x_i)}} \leq \sigma$, where $\sigma \in \mathbb{R}$. 
\end{assumption} 

\begin{assumption}\label{bound_nodes} 
Every node $v_i \in \mathcal{V}$ knows an upper bound
on the number of nodes in the network $n' \geq N$. 
\end{assumption} 

In Assumption~\ref{str_conn} the strong connectivity property facilitates the flow of data among all nodes.  
This is essential for the convergence of distributed algorithms in such network topologies. 
Assumption~\ref{lipsch_str_conv} is standard in first-order distributed optimization problems (see e.g., \cite{2018:Xie, 2018:Li_Quannan}). 
Its properties contribute significantly to the stability and convergence of gradient-based methods.
The smoothness of each local cost function combined with strong convexity, ensures the existence of a global optimal solution for our problem \cite{bubeck2015convex}. 
Specifically, strong convexity guarantees that the global cost function $F$ in \eqref{Global_cost_function} possesses only one minimum, while the Lipschitz smoothness is crucial for establishing convergence rates and for the stability of optimization algorithms in distributed settings. 
Assumption~\ref{bounded_grad} establishes that the stochastic gradient is an unbiased estimator of the true gradient with bounded variance. 
It is crucial for the analysis of convergence rates of stochastic optimization methods in distributed settings. 
Assumption~\ref{bound_nodes} enables nodes to coordinate their estimates and hyperparameters in finite time (see Section~\ref{subsec_FTERC}). 
It can be verified by executing a finite time size estimation algorithm as \cite[Algorithm~$2$]{2023_Rikos_Johan_size_estimation}.

%
%
%
%
\section{Distributed Learning with Automated Stepsizes and Finite Time Coordination}\label{sec:Algorithm1}

In this section we present a distributed algorithm which addresses the optimization problem in \eqref{optim_prob}. 
Our algorithm is named ``Distributed Learning Algorithm with Automated Stepsizes and Real-FTC (DLAS-R-FTC)'' and is detailed as Algorithm~\ref{Algorithm_real} below. 



\subsection{Distributed Learning Algorithm with Automated Stepsizes and Real-FTC (DLAS-R-FTC)}

\begin{algorithm}[!ht]
\caption{DLAS-R-FTC}
\label{Algorithm_real}
\begin{algorithmic}[1]
	\Require For each node $v_i \in \mathcal{V}$ initialize $x_{i}^{[0]}$; set $\eta_{i}^{[0]} = 10^{-10}$; $\theta_{i}^{[0]} = +\infty$; $\lambda_{i}^{[0]} = \eta_{i}^{[0]}$; $\kappa = 0.4$; $\alpha > 0$. Assumptions~\ref{str_conn},~\ref{lipsch_str_conv},~\ref{bounded_grad},~\ref{bound_nodes} hold. 
	\For{$k = 0, 1, \ldots$ each node $v_i$}
            \If{k = 0}
                \CommentState{local update}
                \State 
                \begin{equation}\label{local_stoch_gd}
                x_{i}^{[\frac{1}{2}]} = x_{i}^{[0]} - \lambda_{i}^{[0]} \hat{\nabla} f_{i}(x_{i}^{[0]}); 
                \end{equation}
                \CommentState{coordination}
                \State $x_{i}^{[1]} =$ Algorithm~\ref{Algorithm_real_FTC}($\x^{[\frac{1}{2}]}$), where $\x^{[\frac{1}{2}]} = [(x_{1}^{[\frac{1}{2}]})^\top, \ldots, (x_{N}^{[\frac{1}{2}]})^\top]^\top$; 
                \State $\widetilde{L}_{i}^{[0]} = 0$; 
            \EndIf
            \If{k > 0}
                \CommentState{local update}
                \State $L_{i}^{[k]} = \frac{\norm{\hat{\nabla} f_{i}(x_{i}^{[k]}) - \hat{\nabla} f_{i}(x_{i}^{[k-1]})}}{\norm{x_{i}^{[k]} - x_{i}^{[k-1]}}}$; 
                \State $\widetilde{L}_{i}^{[k]} = (1 - \kappa) L_{i}^{[k]} + \kappa \widetilde{L}_{i}^{[k-1]}$; 
                \State $\eta_{i}^{[k]} = \min \bigl ( \sqrt{1 + \theta_{i}^{[k-1]}} \eta_{i}^{[k-1]} , \frac{\alpha}{\widetilde{L}_{i}^{[k]}})$; 
                \CommentState{coordination}
                \State $\lambda_{i}^{[k]} =$ Algorithm~\ref{Algorithm_real_FTC}($\bm{\eta}^{[k]}$), where \newline $\bm{\eta}^{[k]} = [(\eta_{1}^{[k]})^\top, \ldots, (\eta_{N}^{[k]})^\top]^\top$; 
                \CommentState{local update}
                \State 
                $
                x_{i}^{[k + \frac{1}{2}]} = x_{i}^{[k]} - \lambda_{i}^{[k]} \hat{\nabla} f_{i}(x_{i}^{[k]}); 
                $
                \CommentState{coordination}
                \State 
                $x_{i}^{[k+1]} =$ Algorithm~\ref{Algorithm_real_FTC}($\x^{[k + \frac{1}{2}]}$), where \newline $\x^{[k + \frac{1}{2}]} = [(x_{1}^{[k + \frac{1}{2}]})^\top, \ldots, (x_{N}^{[k + \frac{1}{2}]})^\top]^\top$;  
                \CommentState{local update}
                \State $\theta_{i}^{[k]} = \frac{\eta_{i}^{[k]}}{\eta_{i}^{[k-1]}}$;  
            \EndIf
    \EndFor
    \Ensure Each node $v_i \in \mathcal{V}$ calculates the optimal solution $x^*$ of the optimization problem in \eqref{optim_prob}. 
\end{algorithmic}
\end{algorithm}

\begin{algorithm}[!ht]
\caption{Distributed FTERC} 
\label{Algorithm_real_FTC} 
\begin{algorithmic}[2]
	\Require Each node $v_i \in \mathcal{V}$ initializes $y_i^{[0]}$
        \Ensure Each node $v_i \in \mathcal{V}$ calculates $\frac{1}{N} \sum_{i=1}^N y_i^{[0]}$
	\For{$k = 0, 1, \ldots$ each node $v_i$ does}
            \If{k = 0} 
                \State execute FTERC in Section~\ref{subsec_FTERC} for $2n'$ time steps to compute $y_i^{[1]}$, and also determine $M_i$, and $\beta_i$ 
            \EndIf
            \If{k = 1} 
                \State execute max-consensus for $D$ time steps with input $M_i + 1$ to determine $M_{\max}$
                \State execute ratio consensus in \eqref{eq:1} with input $y_i^{[1]}$ for $n'$ time steps to determine $y_i^{[2]}$ 
            \EndIf
            \If{k > 1} 
                \State execute ratio consensus in \eqref{eq:1} with input $y_i^{[k]}$ for $t_{\max} = M_{\max} + 1$ time steps to determine $y_i^{[k+1]}$ 
            \EndIf
    \EndFor
\end{algorithmic}
\end{algorithm}

The intuition of Algorithm~\ref{Algorithm_real} is the following. 
During the first time step $k=0$, nodes perform a local update using stochastic gradient descent (see step~$3$). 
Then they participate in a coordination phase to align their estimates of the optimal solution. 
This coordination is achieved by executing Algorithm~\ref{Algorithm_real_FTC} until convergence (see step~$4$). 
Then, for each time steps $k>0$ nodes estimate their local Lipschitz constant (see step~$8$). 
They filter this estimate for smoothing it, and use it to calculate an appropriate step size (see step~$9$, $10$). 
They then agree on a common step size across the network through another coordination phase by executing Algorithm~\ref{Algorithm_real_FTC} until convergence (see step~$11$). 
Finally, they perform a local update using stochastic gradient descent (see step~$12$). 
They execute Algorithm~\ref{Algorithm_real_FTC} until convergence to align their estimates, they update their parameters, and repeat the operation (see step~$13$, $14$). 

Algorithm~\ref{Algorithm_real_FTC} describes the Finite-Time Exact Ratio Consensus (FTERC), which enables nodes to calculate the \textit{exact} average of their states in a finite number of steps in a distributed manner. 
The intuition of Algorithm~\ref{Algorithm_real_FTC} is the following. 
During the first iteration (i.e., $k = 0$), nodes execute the FTERC algorithm for a specific number of time steps to compute an updated value of their state and also determine local parameters (see steps~$3$, $4$). 
In the second iteration ($k = 1$), nodes employ a max-consensus algorithm for $D$ time steps to identify a maximum parameter value across the network. 
Then, the ratio consensus method is applied to for $n'$ time steps to enable nodes to update their state values. 
For all subsequent iterations ($k > 1$), nodes continue to apply the ratio consensus method for $t_{\max}$ time steps to update their states. 

\begin{remark}[Comparison with Literature]
The operation of our Algorithm~\ref{Algorithm_real} (DLAS-R-FTC) builds upon recent advancements in distributed optimization by incorporating an automated stepsize selection mechanism that adapts each node's stepsize based on its locally stored data (or local objective function). 
Additionally, it employs a finite-time coordination algorithm (Algorithm~\ref{Algorithm_real_FTC}) to eliminate stepsize heterogeneity among nodes and also align their states, thereby enhancing convergence speed and stability.
This approach represents a significant departure from existing works in the literature. 
For instance, algorithms such as \cite{liggett2022distributed, Yongqiang2024AutomatedStepsizes} rely on dynamic consensus operations that involve only a single communication step between nodes to reduce stepsize heterogeneity. 
While this strategy reduces heterogeneity to some extent, it does not eliminate it entirely, leaving room for convergence instabilities and potential divergence. 
In contrast, our Algorithm~\ref{Algorithm_real} (DLAS-R-FTC) leverages finite-time coordination to fully eliminate stepsize heterogeneity, ensuring synchronized updates across nodes and significantly improving convergence stability.
Furthermore, our algorithm aligns the states of nodes onto the consensus subspace through by implementing (Algorithm~\ref{Algorithm_real_FTC} between gradient descent updates. 
This design choice is inspired by works such as \cite{2023:Rikos_Johan_IFAC}. 
This approach enables faster and more stable convergence compared to existing methods (e.g., see \cite[Fig.~$1$]{2023:Bastianello_Rikos_Johansson}). 
Overall, by achieving exact synchronization of stepsizes and state alignment in finite time, our algorithm outperforms prior techniques in both convergence speed and accuracy, as we also show in Section~\ref{sec:simulation}. 
\end{remark}

\subsection{Convergence Analysis of DLAS-R-FTC}

We now analyze the convergence of Algorithm~\ref{Algorithm_real}. 
Initially, we interpret the operation of Algorithm~\ref{Algorithm_real} as an inexact projected gradient descent method. 
This interpretation is important for our subsequent analysis. 
Then, in Theorem~\ref{theorem_convergence_Alg1} we establish the correctness of our algorithm. 
Specifically, we show that as Algorithm~\ref{Algorithm_real} progresses, the state $x_i^{[k]}$ of every node $v_i \in \mathcal{V}$ converges to the optimal solution $x^*$, thereby addressing the optimization problem \eqref{optim_prob}. 

Firstly, let us note that Algorithm~\ref{Algorithm_real} can be interpreted as an inexact projected gradient descent method as follows 
\begin{equation}\label{eq:algorithm-inexact} 
	\x^{[k+1]} = \proj_\C\left( \x^{[k]} - \lambda_{i}^{[k]} \nabla f(\x^{[k]}) \right) + \e^{[k]}_g. 
\end{equation} 
The vector $\e^{[k]}_g$ accounts for the inexactness introduced due to the usage of stochastic gradients (see Assumption~\ref{bounded_grad}) and is defined as 
    \begin{align} 
             \e^{[k]}_g := & \proj_\C(\x^{[k]} - \lambda_{i}^{[k]} \hat{\nabla} f(\x^{[k]})) \nonumber \\
             & - \proj_\C(\x^{[k]} - \lambda_{i}^{[k]} \nabla f(\x^{[k]})) , \label{eq:error-stoch_grad}
    \end{align}
where $\hat{\nabla} f(\x^{[k]})$ is the stochastic gradient of $f(\x^{[k]})$, and $\nabla f(\x^{[k]})$ is the true gradient of $f(\x^{[k]})$. 

We are now ready to establish correctness of Algorithm~\ref{Algorithm_real} via the following theorem. 

\begin{theorem}\label{theorem_convergence_Alg1} 
    Under Assumptions~\ref{str_conn}--\ref{bound_nodes} Algorithm~\ref{Algorithm_real} generates a sequence of points $\{\x^{[k]}\}$ which satisfy 
        \begin{align} 
             \ev{\norm{\x^{[k+1]} - \x^*}} \leq & \ \left( \prod_{t = 0}^k \zeta^{[t]} \right) \ \ev{\norm{\x^{[0]} - \x^*}}  \nonumber\\
            + &\sigma \sum_{h = 0}^{k} \lambda_{i}^{[h]} \prod_{j = h+1}^{k} \zeta^{[j]} , \label{convergence_bound_alg1}
        \end{align}
    where $\zeta^{[t]} = \max\{ |1 - \lambda_{i}^{[t]} L|, |1 - \lambda_{i}^{[t]} \mu| \} \in (0, 1)$, $L$ and $\mu$ are defined in Assumption~\ref{lipsch_str_conv} and $\sigma$ is defined in Assumption~\ref{bounded_grad}, respectively. 
\end{theorem}

\begin{proof} 
See Appendix~\ref{appen_theorem_convergence_Alg1}.
\end{proof} 

\begin{remark}[Alternative Bound of Theorem~\ref{theorem_convergence_Alg1}]\label{alter_bound_alg1}
    In our analysis in Theorem~\ref{theorem_convergence_Alg1} we established a bound that depends on products of the quantities $\zeta^{[k]}$ and $\lambda^{[k]}$. 
    However, we can simplify the expression in \eqref{convergence_bound_alg1} by considering $\nu^{[k]} = \max_{1 \leq t \leq k} \{ \zeta^{[t]} \} = \max_{1 \leq t \leq k} \{ \max\{ |1 - \lambda_{i}^{[t]} L|, |1 - \lambda_{i}^{[t]} \mu| \} \} \in (0, 1)$, and $\overline{\lambda}_{i}^{[k]} = \max_{1 \leq t \leq k} \{ \lambda_{i}^{[t]} \}$. 
    Considering $\nu^{[k]}$, and $\overline{\lambda}_{i}^{[k]}$, we have that under Assumptions~\ref{str_conn}--\ref{bound_nodes} Algorithm~\ref{Algorithm_real} generates a sequence of points $\{\x^{[k]}\}$ which satisfy 
        \begin{align} 
             \ev{\norm{\x^{[k+1]} - \x^*}} \leq & \ (\nu^{[k]})^k \ \ev{\norm{\x^{[0]} - \x^*}}  \nonumber\\
            + &\left( \overline{\lambda}_{i}^{[k]} \sigma \right) \frac{1 - (\nu^{[k]})^{k+1}}{1 - \nu^{[k]}}. \label{convergence_bound_alg1_alter}
        \end{align}
    The updated expression in \eqref{convergence_bound_alg1_alter} is more compact (compared to \eqref{convergence_bound_alg1}) but yields a more conservative bound regarding the convergence of Algorithm~\ref{Algorithm_real}. 
\end{remark} 



\section{Simulation Results}\label{sec:simulation}


In this section, we demonstrate the correctness of Algorithm~\ref{Algorithm_real} (DLAS-R-FTC) through numerical simulations. 

We focus on the \textit{linear regression} problem characterized by problem \eqref{optim_prob} with local cost function of the form 
\begin{equation}\label{local_func_simulations_1} 
    f_i(x_i) = \frac{1}{m_i} \sum_{h=1}^{m_i} \norm{ x_i - \xi_i^h }^2 , 
\end{equation} 
for each node $v_i \in \mathcal{V}$. 
In~\eqref{local_func_simulations_1}, each node $v_i$ stores $m_i = 50$ data samples. 
For each data sample $\xi_i^h = [ \chi_i^h \ \psi_i^h ]$ of each node $v_i \in \mathcal{V}$ (where $h \in \{ 1, 2, ..., m_i \}$), the parameter $\chi_i^h$ was uniformly sampled from the set $[-5, 5]$, and the parameter $\psi_i^h$ was generated as $\psi_i^h = \beta + \vartheta \chi_i^h + \gamma$, with  $\beta = 4$, $\vartheta = 3$ being the regression parameter (unknown to the nodes), and $\gamma \sim \mathcal{N}(0, 7^2)$ is zero-mean i.i.d. Gaussian noise with standard deviation $7$ and variance $7^2$. 
Note that at each iteration $k$, every node $v_i$ computes a stochastic gradient $\hat{\nabla} f_{i}(x_{i}^{[k]})$ using \textit{a single} uniformly sampled data point $\xi_i^h$ from its local dataset (as described in Assumption~\ref{bounded_grad}). 
The distribution of our data samples is shown in Fig.~\ref{data_samples_plot} with the solid line representing the noise-free linear model with parameters $\beta, \vartheta$ that was used to generate the data. 

\begin{figure}[t]
	\centering
\includegraphics[width=0.95\linewidth]{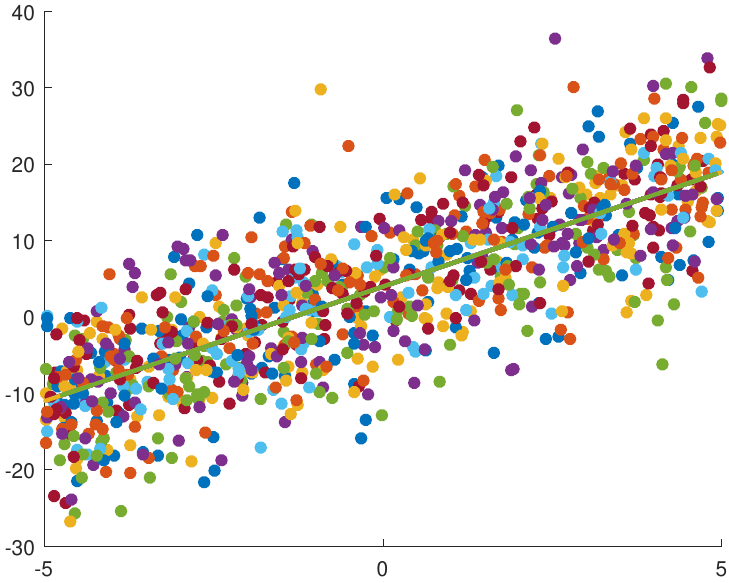}
	\caption{Distribution of data samples across all nodes. The solid line represents the noise-free linear model.} 
	\label{data_samples_plot}
\end{figure}

In Fig.~\ref{convergence_plot} we demonstrate the operation of our Algorithm~\ref{Algorithm_real} (DLAS-R-FTC) against \cite[Algorithm~$3$ (DLAS)]{liggett2022distributed}. 
We show the error $\varepsilon^{[k]}$ defined as 
\begin{equation}\label{eq:distance_optimal} 
    \varepsilon^{[k]} = \sqrt{ \sum_{i=1}^N \frac{(x_i^{[k]} - \vartheta)^2}{(x_i^{[0]} - \vartheta)^2} } , 
\end{equation} 
averaged over $10$ randomly generated digraphs $\mathcal{G}_d = (\mathcal{V}, \mathcal{E})$ with $N = 20$ nodes each. 
Assumptions~\ref{str_conn},~\ref{lipsch_str_conv},~\ref{bounded_grad},~\ref{bound_nodes} hold for each randomly generated digraph. 
The error $\varepsilon^{[k]}$ is plotted on a logarithmic scale against the number of iterations $k$ with $\vartheta$ being the optimal solution of our optimization problem in \eqref{optim_prob}. 

\begin{figure}[t]
	\centering
\includegraphics[width=0.95\linewidth]{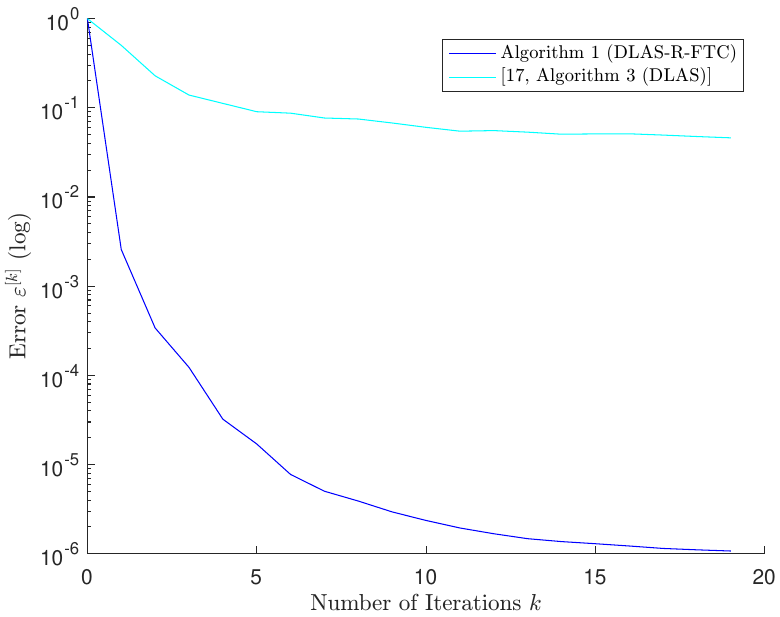}
	\caption{Convergence comparison of Algorithm~\ref{Algorithm_real} (DLAS-R-FTC) against \cite[Algorithm~$3$ (DLAS)]{liggett2022distributed} for the linear regression problem. 
    The normalized error $\varepsilon^{[k]}$ (defined in \eqref{eq:distance_optimal}) is averaged over $10$ random digraphs with $N = 20$ nodes and plotted against iterations $k$ on a logarithmic scale.} 
	\label{convergence_plot}
\end{figure}

In Fig.~\ref{convergence_plot} we can see that Algorithm~\ref{Algorithm_real} (DLAS-R-FTC) outperforms \cite[Algorithm~$3$ (DLAS)]{liggett2022distributed} in terms of convergence speed. 
This improvement is mainly due to to the implementation of Algorithm~\ref{Algorithm_real_FTC} in (i) step-$11$, and (ii) step-$13$ of Algorithm~\ref{Algorithm_real} (DLAS-R-FTC). 
The implementation in step-$11$ enables nodes to utilize equal step-sizes eliminating stepsize heterogeneity, and thus improving the convergence speed of the algorithm. 
This ensures that all nodes follow a synchronized descent direction, avoiding the misalignment caused by heterogeneous stepsizes. 
The implementation in step-$13$ allows nodes to achieve a more accurate estimation of the optimal solution at each time step $k$, thereby improving the convergence speed. 
This approach outperforms most existing methods, which rely on a single step of local data processing and one step of coordination among nodes, as opposed to employing a finite-time algorithm like Algorithm~\ref{Algorithm_real_FTC} for node coordination (e.g., see \cite[Fig.~$1$]{2023:Bastianello_Rikos_Johansson}).
Additionally, the implementation in iteration step-$13$ enables nodes to coordinate their updated states onto the consensus subspace $C$ and thus fulfill \eqref{constr_same_x}. 
In Fig.~\ref{convergence_plot}, we can also see that during the operation of Algorithm~\ref{Algorithm_real} (DLAS-R-FTC) the error $\varepsilon^{[k]}$ converges to a neighborhood of the optimal solution (a value almost equal to $10^{-6}$). 
This behavior aligns with our results in Theorem~\ref{theorem_convergence_Alg1}. 
Specifically, 
in Theorem~\ref{theorem_convergence_Alg1} 
we have that the state of each node converges to a neighborhood of the optimal solution due to nodes utilizing stochastic gradients (see Assumption~\ref{bounded_grad}). 
This neighborhood depends on (i) the bounded variance of stochastic gradients (since from Assumption~\ref{bounded_grad} we have that for every node $v_i$ it holds $\ev{\norm{\hat{\nabla} f_i (x_i) - \nabla f_i (x_i)}} \leq \sigma$, where $\sigma \in \mathbb{R}$), and (ii) the sequence of step-sizes $\{ \lambda_{i}^{[k]} \}$. 




%
%
%
%
\section{Conclusions and Future Directions}\label{sec:conclusions}

\subsection{Conclusions}

In our paper, we focused on the distributed learning problem by introducing a novel algorithm (named DLAS-R-FTC) which incorporates a mechanism for automating stepsize selection among nodes. 
This mechanism relies on the implementation of a finite time coordination protocol for eliminating stepsize heterogeneity among nodes. 
We analyzed the operation of our algorithm and established its convergence to the optimal solution for the case where the local cost functions of nodes are strongly convex and smooth. 
We concluded our paper with numerical simulations for a linear regression problem, showcasing that our algorithm outperforms current approaches in terms of convergence speed and accuracy. 

\subsection{Future Directions}

This work reveals new challenges in adaptive coordination mechanisms and optimization strategies for distributed learning systems. Indicatively, some possible directions are the following. 
\begin{list4}
\item Part of ongoing research focuses on extending our algorithm to distributed learning problems where the local cost functions are convex but not strongly convex. 
\item Additionally, we plan to consider more realistic communication channels in which the bandwidth is limited and introduce resource efficient information exchange mechanisms. 
\item Finally, we plan to consider wireless communication channels in which there are also packet drops. In such setups, finite-time algorithms, such as FTERC, fail. Thus, more robust coordination mechanisms~\cite{CHARALAMBOUS2024} are required.
\end{list4} 



\appendices
%
%
%
%
\section{
Proof of Theorem~\ref{theorem_convergence_Alg1}
}

\label{appen_theorem_convergence_Alg1}
\noindent Eq.~\eqref{eq:algorithm-inexact} can be written as 
\begin{equation}\label{inex_eq_2}
    \x^{[k+1]} - \x^{*} = \proj_\C\left( \x^{[k]} - \lambda_{i}^{[k]} \nabla f(\x^{[k]}) \right) - \x^{*} + \e^{[k]}_g . 
\end{equation}
Taking the norm of \eqref{inex_eq_2} and using the triangle inequality we have 
\begin{align}
            \norm{\x^{[k+1]} - \x^{*}} \leq & \ \norm{\proj_\C\left( \x^{[k]} - \lambda_{i}^{[k]} \nabla f(\x^{[k]}) \right) - \x^{*}} \nonumber\\
            & + \norm{\e^{[k]}_g}. \label{inex_eq_3}
\end{align}
Besides \eqref{inex_eq_3} we also have that
\begin{align}
            &\norm{\proj_\C\left( \x^{[k]} - \lambda_{i}^{[k]} \nabla f(\x^{[k]}) \right) - \x^{*}} + \norm{\e^{[k]}_g}  \nonumber \\
          & \leq  \ \zeta^{[k]} \norm{\x^{[k]} - \x^{*}} + \norm{\e^{[k]}_g}, \label{inex_eq_4}
\end{align}
since from Assumption~\ref{lipsch_str_conv} our problem in \eqref{optim_prob} is smooth and strongly convex with smoothness constant $L$ and strong convexity constant $\mu$ \cite{taylor_exact_2018}, and also $$
\x^{*} = \proj_\C\left( \x^{*} - \lambda_{i}^{[k]} \nabla f(\x^{*}) \right).
$$
Combining \eqref{inex_eq_3} and \eqref{inex_eq_4}, and taking the expected value we have 
\begin{equation}\label{inex_eq_5}
    \ev{\norm{\x^{[k+1]} - \x^{*}}} \leq \zeta^{[k]} \ev{\norm{\x^{[k]} - \x^{*}}} + \ev{\norm{\e^{[k]}_g}}. 
\end{equation}
Focusing on \eqref{inex_eq_5} and \eqref{eq:error-stoch_grad} we bound the term $\ev{\norm{\e^{[k]}_g}}$ as follows 
\begin{align}
            &\ev{\norm{\e^{[k]}_g}}  \nonumber \\
            &\overset{(i)}{\leq}  \ev{\norm{\x^{[k]} - \lambda_{i}^{[k]} \hat{\nabla} f(\x^{[k]}) - (\x^{[k]} - \lambda_{i}^{[k]} \nabla f(\x^{[k]}))}} \ \nonumber \\ 
            &=  \lambda_{i}^{[k]} \ev{\norm{\hat{\nabla} f(\x^{[k]}) - \nabla f(\x^{[k]}))}} \ \nonumber \\ 
            &\leq  \lambda_{i}^{[k]} \sigma , \label{inex_eq_6}
\end{align} 
where the inequality $(i)$ holds from \cite[Proposition~$4.16$]{bauschke_convex_2017} (i.e., the projection is $1$-Lipschitz continuous). 
To complete the proof, we iterate the inequality in \eqref{inex_eq_5} combined with \eqref{inex_eq_6} and apply the properties of geometric series. 

\bibliographystyle{IEEEtran}
\bibliography{references}

\end{document}